\documentclass[pra,showpacs,twocolumn,superscriptaddress]{revtex4}
\newcommand{\FigDirectory}{.}

\usepackage{amsmath}                          
\usepackage{amsfonts}                         
\usepackage{amssymb}                          
\usepackage{amscd}                            
\usepackage{amsthm}                           

\usepackage{dsfont}                           
\usepackage[usenames]{color}                  

\usepackage{graphicx}                         
\usepackage{epstopdf}                         
\usepackage{subfigure}			              
\usepackage{comment}

\usepackage[colorlinks=true,
            urlcolor=blue,
            citecolor=blue,
            linkcolor=black,
            pdfstartview=FitH,
            pdfpagemode=None]{hyperref}

\usepackage{datetime}                         
\usepackage{longtable}
\usepackage{multirow}  
\usepackage{algorithm}
\usepackage{algorithmic}
\usepackage{enumerate}

\renewcommand{\>}{\rangle}

\newcommand{\lsim}{\mathrel{\raise.3ex\hbox{$<$\kern-.75em\lower1ex\hbox{$\sim$}}}}
\newcommand{\gsim}{\mathrel{\raise.3ex\hbox{$>$\kern-.75em\lower1ex\hbox{$\sim$}}}}

\def\QECCnk[[#1,#2]]{[\![#1, #2]\!]}

\def\QECCnkq[[#1,#2,#3]]{[\![#1, #2]\!]_{#3}^{\vphantom{T}}}

\def\QECCnkd[[#1,#2,#3]]{[\![#1, #2, #3]\!]}

\def\QECCnkdq[[#1,#2,#3,#4]]{[\![#1, #2, #3]\!]_{#4}^{\vphantom{T}}}

\def\QECCnkgd[[#1,#2,#3,#4]]{[\![#1, #2, #3, #4]\!]}

\def\QECCnkgdq[[#1,#2,#3,#4,#5]]{[\![#1, #2, #3, #4]\!]_{#5}^{\vphantom{T}}}

\def\QECCnkdc[[#1,#2,#3,#4]]{[\![#1, #2, #3; #4]\!]}

\def\QECCnkdcq[[#1,#2,#3,#4,#5]]{[\![#1, #2, #3; #4]\!]_{#5}^{\vphantom{T}}}

\def\QECCnkgdcq[[#1,#2,#3,#4,#5,#6]]{%
  [\![#1, #2, #3, #4; #5]\!]_{#6}^{\vphantom{T}}}

\newcommand{\bigO}{{\cal O}}

\def\openone{\leavevmode\hbox{\small1\normalsize\kern-.33em1}}

\newcommand{\calL}{{\mathcal{L}}}   
\newcommand{\calS}{{\mathcal{S}}}   
\newcommand{\calZ}{{\mathcal{Z}}}   

\newcommand{\etal}{\textit{et~al.}}

\newcommand{\ie}{\textit{i.e.}}

\newcommand{\eg}{\textit{e.g.}}

\long\def\symbolfootnote[#1]#2{\begingroup%
\def\thefootnote{\fnsymbol{footnote}}\footnote[#1]{#2}\endgroup}

\input{Qcircuit.tex}

\theoremstyle{definition}  
\newtheorem{definition}{Definition}
\newtheorem{theorem}{Theorem}
\newtheorem{lemma}{Lemma}

\makeatletter
\def\grabtimezone #1#2#3#4#5#6#7#8#9{\grabtimezoneB}
\def\grabtimezoneB #1#2#3#4#5#6#7{\grabtimezoneC}
\def\grabtimezoneC #1#2'#3'{#1#2#3 UTC}
\def\timezone{\expandafter\grabtimezone\pdfcreationdate}
\makeatother

\begin{document}


%
\title{Complex instruction set computing architecture\\ for performing
accurate quantum $Z$ rotations with less magic}

\author{Andrew J. \surname{Landahl}}
\email[]{alandahl@sandia.gov}
\affiliation{Advanced Device Technologies,
             Sandia National Laboratories,
             Albuquerque, NM, 87185, USA}
\affiliation{Center for Quantum Information and Control,
             University of New Mexico,
             Albuquerque, NM, 87131, USA}
\affiliation{Department of Physics and Astronomy,
             University of New Mexico,
             Albuquerque, NM, 87131, USA}
\author{Chris \surname{Cesare}}
\email[]{ccesare@unm.edu}
\affiliation{Center for Quantum Information and Control,
             University of New Mexico,
             Albuquerque, NM, 87131, USA}
\affiliation{Department of Physics and Astronomy,
             University of New Mexico,
             Albuquerque, NM, 87131, USA}



\begin{abstract}

We present quantum protocols for executing arbitrarily accurate $\pi/2^k$
rotations of a qubit about its $Z$ axis.  Reduced instruction set computing
(\textsc{risc}) architectures typically restrict the instruction set to
stabilizer operations and a single non-stabilizer operation, such as
preparation of a ``magic'' state from which $T = Z(\pi/4)$ gates can be
teleported.  Although the overhead required to distill high-fidelity copies
of this magic state is high, the subsequent quantum compiling overhead to
realize arbitrary $Z$ rotations in a \textsc{risc} architecture can be much
greater.  We develop a complex instruction set computing (\textsc{cisc})
architecture whose instruction set includes stabilizer operations and
preparation of magic states from which $Z(\pi/2^k)$ gates can be teleported,
for $2 \leq k \leq k_{\text{max}}$.  This results in a reduction in the
resources required to achieve a desired gate accuracy for $Z$ rotations.
The key to our construction is a family of shortened quantum Reed-Muller
codes of length $2^{k+2}-1$, whose magic-state distillation threshold
shrinks with $k$ but is greater than $0.85\%$ for $k \leq 6$.

\end{abstract}

%
\maketitle


%
\section{Introduction}

One of the biggest challenges in quantum information science is that quantum
information is incredibly fragile.  Even with great experimental care,
decoherence can quickly corrupt key features such as superposition and
entanglement.  To circumvent the ravages of decoherence, one can consider
alternative models of quantum computation, such as adiabatic quantum
computation \cite{Farhi:1998a, Aharonov:2004a, Mizel:2006a}, which may offer
direct physical immunity to certain classes of noise \cite{Childs:2001a,
Aberg:2004a, Aberg:2005a, Sarandy:2005a, Roland:2005a, Ashhab:2006a,
Gaitan:2006a, Tiersch:2006a, Amin:2007a, Amin:2008a, deVega:2010a}.  Another
approach is to encode quantum information redundantly in an error-correcting
code and process it fault-tolerantly to suppress the catastrophic
propagation of errors \cite{Shor:1995a, Shor:1996a}.  Somewhat miraculously,
this latter approach works, and works arbitrarily well, when quantum
computations are expressed as quantum circuits in which each elementary
operation has a failure probability below a value known as the
\emph{accuracy threshold} \cite{Aharonov:1997a, Aharonov:1999a,
Kitaev:1997b, Steane:1997a, Knill:1998a, Preskill:1998a, Preskill:1998c}.
Estimates for the accuracy threshold vary, and depend in part on the
specifics of the fault-tolerant quantum computing protocol used.  One of the
more favorable estimates is $\approx 1\%$ for a protocol based on Kitaev's
surface codes \cite{Kitaev:1996a, Dennis:2002a, Raussendorf:2007b,
Fowler:2008a}.  An outstanding grand challenge in quantum information
science is finding a way to marry fault-tolerance methods with intrinsically
robust computational models to achieve fault tolerance with more achievable
resource requirements \cite{Jordan:2006a, Lidar:2007a, Paz-Silva:2012a,
Young:2012a}.

One of the factors driving up the resource requirements in fault-tolerant
quantum computing is the need to restrict the set of elementary operations
in the ``primitive'' or ``physical'' instruction set to be finite.  This is
necessary because these instructions are presumed to be implementable only
up to some maximal accuracy.  One of the main jobs of a fault-tolerant
quantum computing protocol is to define how one should sequence these
primitive instructions together to synthesize arbitrarily accurate versions
of each element of a universal ``encoded'' or ``logical'' instruction set,
even when the primitive instructions themselves are faulty.  Then, using
these logical instructions, one can realize any quantum algorithm
arbitrarily reliably, even in the face of decoherence and other sources of
noise.

In a typical fault-tolerant quantum computing protocol, some logical
instructions are ``easy'' to synthesize in that their error is solely a
function of the errors in the primitive instructions from which they are
composed.  The accuracy of these logical instructions can be improved
arbitrarily well by using arbitrarily good quantum codes.  More
quantitatively, the number of gates and qubits required to achieve
approximation error $\epsilon$ for the ``easy'' instructions scales as
$\bigO(\log^\alpha(1/\epsilon))$, where $\alpha$ depends on the protocol,
predominantly on the quantum code and classical decoding algorithm it uses.
Standard techniques for realizing such gates include transversal action
\cite{Preskill:1998a, Preskill:1998c} and code deformation
\cite{Dennis:2002a, Raussendorf:2007b}.  2D topological codes using
most-likely-error decoding can achieve $\alpha = 3$ \cite{Dennis:2002a,
Raussendorf:2007b}; Pippenger has conjectured that it should be possible to
lower $\alpha$ all the way to $1$ \cite{Ahn:2004a}.

Most protocols also have a set of logical instructions that are ``hard'' to
synthesize, requiring additional methods and resources.  The Eastin-Knill
theorem, for example, guarantees that no protocol can realize a universal
logical instruction set by transversal action alone \cite{Eastin:2008a}.  A
typical approach to synthesizing these hard logical instructions is to use
the ``magic state'' approach, in which the ``hard'' instructions are state
preparations that are distilled to high fidelity using the ``easy''
operations \cite{Bravyi:2005a}.  The number of ideal gates and qubits
required to achieve approximation error $\epsilon$ in this approach scales
as $\bigO(\log^\beta(1/\epsilon))$, where $\beta$ depends on the magic-state
distillation protocol.  When the the resource costs for the ``easy'' gates
are also considered, the combined overhead scales as $\bigO(\log^{\alpha +
\beta}(1/\epsilon))$.  In the well-studied Bravyi-Kitaev 15-to-1 distillation
protocol \cite{Bravyi:2005a}, $\beta = \log_3 15 \approx 2.47$.  More recent
constructions by Bravyi and Haah \cite{Bravyi:2012a} and by Jones
\cite{Jones:2012a} achieve $\beta = \log_2 3 \approx 1.58$.  Bravyi and Haah
conjecture that it should be possible to lower $\beta$ all the way to $1$
\cite{Bravyi:2012a}.

As an aside, it is worth mentioning that fault-tolerant quantum computing
protocols based on some quantum codes have no ``hard'' logical instructions
at all.  For example, the 3D (and higher-dimensional) topological color
codes have this feature \cite{Bombin:2007a, Landahl:2011a}.  They cleverly
circumvent the Eastin-Knill theorem by making (non-transversal!) quantum
error correction be the process by which magic-states are prepared.  A
challenge to using these codes in practice is that implementing them without
relying on long-distance quantum communication requires 3D spatial geometry,
but many quantum technologies are naturally restricted to 1D or 2D.  Even
more challenging is that the only explicit 3D color code of which we are
aware is the 15-qubit shortened quantum Reed-Muller code
\cite{Bombin:2007a}. Concatenated schemes using the 15-qubit code would lead
to a fault-tolerant scheme with only ``easy'' instructions, but concatenated
schemes typically suffer significant performance losses when realized in a
fixed spatial dimension.  For example, the largest accuracy threshold of
which we are aware for a concatenated-coding protocol in a semiregular 2D
geometry is $1.3 \times 10^{-5}$ \cite{Spedalieri:2009a}.

Because of the additional overhead incurred in synthesizing ``hard'' logical
instructions, research to date has focused on what one might term
\emph{reduced instruction set computing}, or \textsc{risc}, architectures in
which only a single ``hard'' logical instruction is added to an otherwise
``easy'' logical instruction set.  However, while a \textsc{risc}
architecture minimizes the number of hard instructions in an instruction
set, it does not necessarily minimize the number of hard instructions used
in specific algorithms.  For example, in order to compile the logical
instructions into a sequence that approximates a quantum computation with
error at most $\epsilon$, one must use $\bigO(\log^\gamma(1/\epsilon))$
gates, where $\gamma$ depends on the quantum compiling algorithm used.  The
overall cost of fault-tolerantly implementing a quantum computation is then
$\bigO(\log^{\alpha + \beta + \gamma}(1/\epsilon))$.  By increasing the size
of the instruction set so that one has a \emph{complex instruction set
computing}, or \textsc{cisc} architecture, one can optimize both $\beta$ and
$\gamma$ together rather than separately.  When quantum compiling is
optimized independently, $\gamma$ can be no lower than $1$
\cite{Harrow:2002a}, a value recently achieved by an explicit
Diophantine-equation-based algorithm by Selinger \cite{Selinger:2012a} and
Kliuchnikov \etal\ \cite{Kliuchnikov:2012b}.  For comparison's sake, the
more well-studied Dawson-Nielsen variant of the Solovay-Kitaev algorithm
achieves $\gamma = \log 5/\log(3/2) \approx 3.97$ \cite{Dawson:2005a}.

To compare and contrast the \textsc{risc} and \textsc{cisc} approaches more
concretely without being encumbered by details of quantum error correcting
codes and fault tolerance (which only contribute to $\alpha$ and a
delineation of which logical instructions are ``easy'' or
``hard''---properties shared by both approaches), we abstract these details
away and simply consider the straightforward problem of how to approximate
$\pi/2^k$ rotations of a qubit about its $Z$ axis with a desired error at
most $\epsilon'$ when we are given the ability to perform a proscribed set
of ``easy'' instructions that are error-free and a proscribed set of
``hard'' instructions that have error at most $\epsilon > \epsilon'$.  In
this setting, it is clear that some kind of distillation of the hard
instructions will be necessary to synthesize the $Z$ rotations with lower
error.  $Z(\pi/2^k)$ rotations are a natural candidate transformation to use
to compare \textsc{risc} and \textsc{cisc} approaches, because they arise in
many quantum algorithms, for example those that make use of the quantum
Fourier transform \cite{Nielsen:2000a}.

In Sec.~\ref{sec:problem}, we formulate the statement of the problem we are
considering more precisely.  In Sec.~\ref{sec:RISC}, we review the standard
\textsc{risc} solution to this problem.  In Sec.~\ref{sec:CISC}, we describe
our \textsc{cisc} solution, and compare it to the \textsc{risc} solution,
demonstrating that for a regime of target $\epsilon'$ our solution offers a
reduction in the number of resource states used to achieve this task.
Sec.~\ref{sec:conclusion} concludes.  Appendix
\ref{sec:Quantum-Reed-Muller-codes} elaborates the shortened quantum
Reed-Muller codes we use to effect our protocol, and Appendix
\ref{sec:main-theorem} formulates a testable set of criteria one can use to
check if a code admits $Z(\pi/2^k)$ transversally.

%
\section{Problem statement}
\label{sec:problem}

Consider quantum $Z$ rotations of the form
\begin{align}
Z_k &:=
  \begin{pmatrix} 1 & 0 \\ 0 & e^{i\pi/2^k} \end{pmatrix} 
 =\ e^{i\pi/2^{k+1}} R_z\!\left(\frac{\pi}{2^{k}}\right),
\end{align}
for integers $k \geq 0$.  As a shorthand, we use $Z$ to denote the Pauli
operator $Z_0$ and $S$ and $T$ to denote the rotations $Z_1$ and $Z_2$
respectively.  We are interested in the scenario in which the $Z_k$ gates
are not available directly, but rather their action on $|+\>$ states is,
where $|+\> := H|0\> = (|0\> + |1\>)/\sqrt{2}$ and $H := (X + Z)/\sqrt{2}$.
For concreteness, let $\calZ_{k_\text{max}}$ denote the set of states of the
form
\begin{align}
Z_k|+\> &= \frac{1}{\sqrt{2}}\left(|0\> + e^{i\pi/2^k}|1\>\right)
\end{align}
for $2 \leq k \leq {k_\text{max}}$.

In conjunction with the set $\calS$ of \emph{stabilizer operations}
\cite{Gottesman:1999b}, the set $\calZ_{k_\text{max}}$ can effect universal
quantum computation, even when restricted to ${k_\text{max}} \leq 2$
\cite{Nielsen:2000a}.  Here we restrict our attention to a certain
(overcomplete) generating set for $\calS$, namely the set
consisting of the operations
\begin{align}
\big\{&I,\ X,\ Y,\ Z,\ S,\ S^\dagger,\ H\big\} \cup \big\{|0\>,\ |+\>,\
M_Z,\ M_X \big\} 
\end{align}
and
\begin{align}
    &\big\{ \Lambda(X^{q_1}\otimes \cdots \otimes X^{q_m})\ |\ {q_i \in \{0,
1\}} \big\},
\end{align}
where $I$, $X$, $Y$, and $Z$ denote the Pauli operators, $M_X$ and $M_Z$
denote projective measurements in the $X$ and $Z$ bases (but which may be
``destructive'' in that they do not necessarily prepare $X$ or $Z$
eigenstates after the measurement), and $\Lambda(X^q)$ denotes the
one-control, many-target controlled-\textsc{not} gate, where the number of
targets $m$ is some efficiently computable number.  The unitary gates in
this generating set generate a subgroup of the stabilizer operations known
as the \emph{Clifford} operations \cite{Gottesman:1999b}, which are the set
of operations that conjugate (tensor products of) Pauli operators to (tensor
products of) Pauli operators.

These generators of $\calS$ are ``easy'' to perform at the logical level for
the 4.8.8 2D color codes, motivating our choice \cite{Landahl:2011a}.  The
set is also almost ``easy'' for Kitaev's 2D surface codes
\cite{Kitaev:1996a}, except generating $S$ and $S^\dagger$ requires some
constant startup costs that can be amortized \cite{Anderson:2012a}.
Amazingly, as noted in the introduction, all elements from the set
$\calS \cup  \calZ_2$---a universal set---are ``easy'' to perform at the
logical level for 3D color codes, but 3D geometries are required to realize
error correction with these codes in a spatially local manner
\cite{Landahl:2011a}.
 
While errors in the ``easy'' operations can be suppressed arbitrarily close
to zero by using arbitrarily large 2D topological codes, errors in the
operations in $\calZ_{k_\text{max}}$ cannot, making these operations
``hard'' for these codes.  The states in $\calZ_{k_\text{max}}$ can be
``injected'' into such codes at the logical level \cite{Raussendorf:2007b},
but doing so also injects the errors in the state.  In other words, if the
states in $\calZ_{k_\text{max}}$ have errors that are at most $\epsilon$ (as
measured by the trace distance \cite{Nielsen:2000a}) as primitive
instructions, then the injected states will have errors that are essentially
the same when they become logical instructions, assuming the injection
process itself adds errors at a low enough probability \endnote{How errors
propagate in the injection process is an understudied problem in our
opinion.  However, we will not consider this issue here because we are
abstracting away the details of quantum error correcting codes in our
analysis.}.
 
Motivated by these properties of 2D topological codes, we will fix the
control model for our study to be the aforementioned generators of $\calS$
and $\calZ_{k_\text{max}}$, and the error model to be one in which the
operations in $\calS$ are error-free but in which the $Z_k|+\>$ states in
$\calZ_{k_\text{max}}$ each err by at most $\epsilon$, as measured by the
trace distance.  Notice that this control model makes no reference to codes
or fault-tolerant quantum computing protocols.  We have abstracted these
away to focus on how to combine elementary operations in $\calS$ and
$\calZ_{k_\text{max}}$ to achieve high-fidelity $Z$ rotations.

The question we address here is,

\begin{quote}
 \textit{How many resource states drawn from $\calZ_{k_\text{max}}$ does it
take to approximate $Z_k$ with error at most $\epsilon' < \epsilon$ as a
function of ${k_\text{max}}$, $k$, $\epsilon$, and $\epsilon'$?}
\end{quote}

The values of $k$ we are interested in could be smaller than, equal to, or
larger than $k_{\text{max}}$.  However, since $Z_0$ and $Z_1$ are both in
the error-free set $\calS$, we restrict our attention to $k \geq 2$.

%
\section{Traditional quantum \textsc{risc} architecture solution}
\label{sec:RISC}

The standard method for refining the accuracy of a $Z_k$ rotation is to
synthesize it with what one might term a quantum \emph{reduced instruction
set computing}, or quantum \textsc{risc}, architecture.  The main idea is to
only synthesize $T := Z_2$ gates to high accuracy and then rely on a quantum
compiling algorithm to approximate $Z_k$ arbitrarily well with a quantum
circuit over $T$ gates and adaptive stabilizer operations.  The overall
process can be broken into the three steps of \emph{quantum compiling},
\emph{quantum gate teleportation}, and \emph{magic-state distillation}.

%
\subsection{Protocol}

%
\subsubsection{Quantum compiling}

The first step, \emph{quantum compiling}, generates a classical description
of an ideal quantum circuit that approximates $Z_k$ to accuracy
$\epsilon_{\text{qc}}$ using $\bigO(\log^\gamma(1/\epsilon_{\text{qc}}))$
quantum operations drawn from some instruction set, for some small constant
$\gamma$.  While the error $\epsilon_{\text{qc}}$ can be measured in
multiple ways, a wise choice is to measure $\epsilon_{\text{qc}}$ using the
completely-bounded (``diamond'') trace distance \cite{Kitaev:1997b,
Sacchi:2005a, Watrous:2009a} for reasons that we will explain later.
Examples of quantum compiling algorithms include the Solovay-Kitaev
algorithm \cite{Solovay:2000a, Kitaev:1997b, Nielsen:2000a, Harrow:2002a,
Dawson:2005a, Fowler:2005a, Trung:2012a}, the Kitaev phase kickback
algorithm \cite{Kitaev:1995a, Cleve:1998a, Kitaev:2002a}, programmed ancilla
algorithms \cite{Isailovic:2008a, Jones:2012a, Duclos-Cianci:2012a}, genetic
algorithms \cite{McDonald:2012a}, and even Diophantine-equation algorithms
\cite{Kliuchnikov:2012b, Selinger:2012a}.  When the accuracy demand is not
great, it is sometimes even plausible to use algorithms which take
exponential time to find very short approximation sequences
\cite{Fowler:2011b, Amy:2012a, Kliuchnikov:2012a, Bocharov:2012a}.  As noted
in the introduction, values for $\gamma$ range from $3.97$ to $1$.

Quantum compiling algorithms typically assume that the elements of the
instruction set are error-free.  If one implements the compiled circuit
$Z_k^{(\text{qc})}$ for $Z_k$ with operations that may be in error, the
resulting approximation error will increase.  To calculate the total error
$\epsilon_k$ in this flawed circuit $\tilde{Z}_k^{(\text{qc})}$, we use the
fact that the diamond norm has many useful mathematical properties,
including obeying the triangle inequality, the chaining inequality, and
unitary invariance \cite{Gilchrist:2005a}.  Using these, we can bound
$\epsilon_k$ as
\begin{align}
\epsilon_k  &=
  {d_{\diamond}\!\left(Z_k, \tilde{Z}_k^{(\text{qc})}\right)} \\
  &\leq \
    d_{\diamond}\!\left(Z_k, Z_k^{(\text{qc})}\right)
    + d_{\diamond}\!\left(Z_k^{(\text{qc})},
               \tilde{Z}_k^{(\text{qc})}\right) \\
  &\leq \epsilon_{\text{qc}} +  n_T\epsilon_T,
\end{align}
where the compiled circuit uses $n_T$ $T$ gates, each with error at most
$\epsilon_T$.  To achieve the desired approximation error of $\epsilon'$, it
follows that sufficient conditions are
\begin{align}
\epsilon_{\text{qc}} &\leq C_{\text{qc}}\epsilon' \\
\epsilon_T &\leq C_T\epsilon'/n_T,
\end{align}
for positive constants constrained to obey 
\begin{gather}
C_T + C_{\text{qc}} \leq 1.
\end{gather}

For comparison with our protocol introduced later, we chose the Diophantine
equation-based compiling protocol presented by Selinger in
Ref.~\cite{Selinger:2012a}. This protocol saturates the asymptotic lower
bound on the number of $T$ gates required to approximate a single-qubit
gate, and for $Z$ rotations has a $T$ count of
\begin{align}
n_T\left(\epsilon_{\text{qc}}\right) &\approx 11 +
4\log_2\left(\frac{1}{\epsilon_{\text{qc}}}\right).
\end{align}
%

%
\subsubsection{Quantum gate teleportation}

The second step, \emph{quantum gate teleportation}, replaces each $T$ gate
in the quantum-compiled circuit by an adaptive stabilizer circuit that
teleports the $T$ gate from the state $T|+\>$ or $T^\dagger|+\>$ to the
desired qubit.  An example of a teleportation circuit using $T|+\>$ is
depicted in Fig.~\ref{fig:T-teleport}.  The circuit is also correct if both
$T$ operators are changed to $T^\dagger$; it is even correct if only one of
the $T$ operators is changed to a $T^\dagger$ if the classical control is
also changed to act on a 0 instead of a 1.
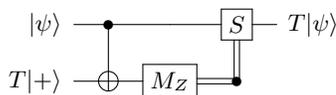
\begin{figure}[htb]
\centerline{
\Qcircuit @C=1em @R=1em {
 \lstick{\ket{\psi}} & \ctrl{1} & \qw & \gate{S} & \qw & \quad T\ket{\psi}
\\
 \lstick{T\ket{+}} & \targ & \gate{M_Z} & \control \cw \cwx  \\
} 
} 
\caption{\small{\label{fig:T-teleport}Circuit for teleporting the $T$ gate
from the $T|+\>$ magic state.}}
\end{figure}

Each teleportation circuit requires the use of just a single $T|+\>$
resource state. The accuracy requirement set by $\epsilon_T$ will determine
whether these are `bare' $T|+\>$ states of accuracy $\epsilon$ or whether
these states are the result of one or more rounds of distillation, described
in the next section.

%
\subsubsection{Magic-state distillation}

The third step, \emph{magic-state distillation}, generates $T|+\>$ or
$T^\dagger|+\>$ states with accuracy $\epsilon_T$ from a much larger
collection of states whose accuracy is only ${\epsilon}$.  Reichardt showed
that this is possible using an \emph{ideal} (error-free) stabilizer circuit
if and only if ${\epsilon}$ is less than the distillation threshold $(2 -
\sqrt{2})/4 \approx 0.146$ \cite{Reichardt:2005a}.  When operations in the
stabilizer circuit can err, the evaluation of the threshold is more complex,
as studied by Jochym-O'Connor \etal\ \cite{Jochym-OConnor:2012a}.

There are multiple variations on how to implement magic-state distillation
discussed in the literature \cite{Knill:2004b, Reichardt:2004a,
Reichardt:2006a, Bravyi:2005a, Meier:2012a, Bravyi:2012a}; a popular one is
the 15-to-1 Bravyi-Kitaev protocol \cite{Bravyi:2005a} based on the 15-qubit
shortened quantum Reed-Muller code $\overline{QRM}(1,4)$. (See Appendix
\ref{sec:Quantum-Reed-Muller-codes} for an explanation of this notation.)

To date, the best distillation scheme in terms of resource costs is a hybrid
of the 15-to-1 Bravyi-Kitaev protocol \cite{Bravyi:2005a}, the 10-to-2
Meier-Eastin-Knill protocol \cite{Meier:2012a}, and the (3$k$+8)-to-$k$
family of protocols discovered by Bravyi and Haah \cite{Bravyi:2012a}.
Bravyi and Haah optimized combinations of these protocols to find the most
efficient way of producing a state $T|+\>$ of target accuracy $\epsilon_T$
\cite{Bravyi:2012a}. The optimization yields about a factor of two
improvement over a scheme which utilizes only a combination of the 15-to-1
and the 10-to-2 protocols. We perform no such optimization over protocols
when we compare to our own distillation protocols, because we already see a
savings of more than an order of magnitude over these. 

We chose to compare our protocol to resource costs incurred by the Selinger
approximation protocol in conjunction with the Meier-Eastin-Knill (MEK)
10-to-2 protocol. For completeness we now provide a brief description of how
the MEK protocol functions \cite{Meier:2012a}. The goal is to prepare a
target resource state, in our case $T|+\>$, with some desired accuracy
$\epsilon_T$ given only faulty copies of the same state with error $\epsilon
> \epsilon_T$. The simplest way to prepare such a state would be to measure
an operator whose eigenstate is $T|+\>$, but given access to only Clifford
operations this cannot be done. To circumvent this problem, more resource
states of accuracy $\epsilon$ are consumed to perform the desired
measurement. For the MEK protocol, the total number of resource states
consumed per round is 10. Additionally, the measurement performed is a
\emph{logical} measurement on an encoded qubit (or qubits). This allows for
the detection of errors during the measurement procedure and is responsible
for the increased accuracy of the output resource states. If the desired
logical measurement outcome has been observed, the syndrome for the code is
measured and, conditioned on the syndrome being error free, running the
decoding circuit leaves two resource states with error $\bigO(\epsilon^2)$.

The code utilized by the MEK protocol is the $[\![4, 2, 2]\!]$ quantum
error-detecting code. The distilled states are the eigenstates of $H$,
denoted $|H\>$, which are related to $T|+\>$ by a Clifford rotation as follows:
\begin{gather}
|H\> = S H \left(T|+\>\right).
\end{gather}
The protocol proceeds as follows:
\begin{enumerate}
\item Encode two (``twirled'') copies of $|H\>$ in the $[\![4, 2, 2]\!]$
code.  ``Twirling'' is performed by the probabilistic process that applies
either $I$ or $H$ to the state, each with probability $1/2$.
\item Perform a measurement of \emph{logical} $H_1 H_2$, which for this code
is the same as transversal $H$ up to a SWAP. This measurement uses eight
additional $|H\>$ states, which can be inferred from the identities in
Fig.~1 of Ref.~\cite{Meier:2012a}.
\item If a $-1$ outcome is obtained for the measurement of $H_1 H_2$, start
over. If a $+1$ outcome is obtained, measure the syndrome for the code.
\item If an error-free syndrome is reported, decode. Otherwise, start over.
\item After decoding there will be two higher fidelity $|H\>$ states with
error $\bigO(\epsilon^2)$.
\end{enumerate} Note that the syndrome measurements can be pushed through
the decoding circuit, becoming single-qubit measurements after decoding is
performed.

Counting the number of resource states required to produce $n_T$ states of
accuracy $\epsilon_T$ is accomplished by numerically evaluating the
recursive relationship
\begin{align}
n_T(\ell) &= 5  n_T(\ell-1) / a(\epsilon_\ell),
\end{align}
where $a(\epsilon_\ell)$ is the probability of the protocol accepting, given
above Eq.~(3) in Ref.~\cite{Meier:2012a}, $\epsilon_\ell$ is the accuracy
after $\ell$ rounds of distillation, and the base of the recursion is simply
$n_T(0) = 5/a(\epsilon)$. Intuitively, this just says that to produce one
resource state of accuracy $\bigO(\epsilon^2)$ requires on average
$5/a(\epsilon)$ states of fidelity $\epsilon$. We use this, in conjunction
with Eq.~(3) in \cite{Meier:2012a} to calculate how many resource states are
required to achieve a target $\epsilon_T$.

%
\subsection{Resource analysis}

As mentioned in the introduction, asymptotically the total number of
operations required to approximate a $Z_k$ gate with error $\epsilon'$ is
$\bigO(\log^{\alpha + \beta + \gamma}(1/\epsilon'))$, where the exponents
describe various overheads of the steps involved: fault-tolerant stabilizer
operations ($\alpha$), magic-state distillation ($\beta$), and quantum
compiling ($\gamma$).  While a good starting point, asymptotic analysis like
this fails to convey the great number of elementary operations needed to
implement $Z_k$ gates, as it sweeps the (large!) constants under the rug.
The explicit expression for the expected number of states used by the
\textsc{risc} approach to approximate $Z_k$ to error $\epsilon'$ using
$T|+\>$ states whose error is $\epsilon$ is
\begin{align}
n_{\text{states}}^{\textsc{risc}}(Z_k, \epsilon', \epsilon)
    &= \left[11+4\log_2\left(\frac{1}{C_{\text{qc}} \epsilon'}\right)\right]
\\
    &\phantom{\quad}\times n_{\text{states}}^{T}\left(\frac{C_T
\epsilon'}{n_T}, \epsilon\right),
    \nonumber
\end{align}
where $n_{\text{states}}^T(C_T \epsilon' / n_T, \epsilon)$ is the number of
$T|+\>$ states of error $\epsilon$ required to produce a $T|+\>$ state of
error $C_T \epsilon'/n_T$. The idea here is to first use the results of
Ref.~\cite{Selinger:2012a} to approximate $Z_k$ to accuracy $C_{\text{qc}}
\epsilon'$, and then replace each $T$ gate in the compiled sequence with a
teleportation circuit using a $T|+\>$ state of accuracy $C_T \epsilon' /
n_T$.

To better appreciate the compiling resources needed, we consider the case
when $C_{\text{qc}} = C_T = 1/2$, which balances the quality demands of
quantum compiling and magic-state distillation.  We give the $T|+\>$
state a generous error rate of $\epsilon = 10^{-4}$, which is well below the
estimated threshold of $\approx 1\%$ for fault-tolerant quantum computation
with surface codes \cite{Raussendorf:2007b, Fowler:2008a}. The number of
states $n_{\text{states}}^{\textsc{risc}}$ required to synthesize $Z_k$
with these parameters to various approximation levels are plotted in the dashed curve in Fig.~\ref{fig:RISC-v-CISC-Zk-plot}. One appealing feature, especially for large values of $k$, is that the curve does not depend on $k$\textemdash the number of states needed is solely a function of the desired output precision.

%
\section{Quantum \textsc{cisc} architecture solution}
\label{sec:CISC}

Now that we've described how to implement $Z_k$ rotations using a quantum
\textsc{risc} architecture, it's natural to ask if extending the instruction
set to a quantum \emph{complex instruction set computing} architecture, or
quantum \textsc{cisc} architecture, could provide any advantage in terms of
a reduction in the required number of resource states.  The point is that in
any given quantum algorithm instance, one isn't interested in applying
\emph{arbitrary} gates but rather a specific set of gates, say $Z_k$ gates
up to some maximum value of $k$ in a quantum Fourier transform.  Because of
this, it may make more sense to just include those gates in the instruction
set to begin with rather than compiling them from a more limited instruction
set.  Even if it is only feasible to include gates up to some value of
$Z_{k_{\text{max}}}$, it is reasonable to expect that the length of the
resulting compiled sequences will be shorter if an arbitrary gate is
required.

%
\subsection{Protocol}

In our protocol we consider a programmed-ancilla \textsc{cisc} architecture,
in which we pre-compile $Z_k|+\>$ states offline that can be used later to
teleport the gate $Z_k$ on demand via the circuit in
Fig.~\ref{fig:Zk-teleportation-circuit}.  While the teleportation may
require a $Z_{k-1}$ gate for correction, iterating this process recursively
is a negative binomial process that converges exponentially quickly---the
expected number of $Z$ rotations for any $k$ is two: $Z_k$ on $|+\>$ and
$Z_{k-1}$ after the measurement.  To achieve error at most $\epsilon'$ on
the teleported $Z_k$ gate, the $Z_k|+\>$ state and the $Z_{k-1}$ gate need
to be performed with errors at most $C_1\epsilon'$ and $C_2\epsilon'$
respectively, where $C_1 + C_2 \leq 1$. 
\begin{figure}[htb]
\centerline{
\Qcircuit @C=1em @R=1em {
 \lstick{\ket{\psi}} & \ctrl{1} & \qw & \gate{Z_{k-1}} & \qw & \quad
Z_k\ket{\psi} \\
 \lstick{Z_k\ket{+}} & \targ & \gate{M_Z} & \control \cw \cwx  \\
} 
} 
\caption{\small{\label{fig:Zk-teleportation-circuit}Magic-state circuit for
teleporting the $Z_k$ gate.}}
\end{figure}
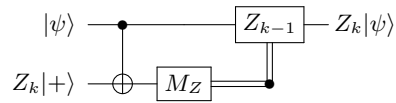

Our \textsc{cisc} approach is distinguished from previous programmed-ancilla
approaches \cite{Isailovic:2008a, Jones:2012a, Duclos-Cianci:2012a} in that
we distill ancilla $Z_k|+\>$ states directly as instructions unto
themselves.  This is a ``top-down'' approach in which some of the time
auxiliary $Z_{k-1}|+\>$ states are needed, and even less of the time
$Z_{k-2}|+\>$ states are needed, and so on, until we get to the point that
very rarely do we need $T|+\>$ states.  The previous approaches are
``bottom-up'' in that they always compile from $T|+\>$ states upwards until
the $Z_{k}$ gate is performed; some of these schemes (notably the recent one
by Duclos-Cianci and Svore \cite{Duclos-Cianci:2012a}) reduce resources by
including intermediate targets, but ultimately they all start from $T|+\>$
preparations at the lowest level.  By starting from the top, we avoid the
need to probe all the way to the bottom most of the time.  As we will see,
this results in savings in the number of operations needed to
synthesize $Z_k$ gates.

The key to our construction is a family of shortened quantum Reed-Muller
codes that are defined in Appendix~\ref{sec:Quantum-Reed-Muller-codes}.  The
property of these codes that we harness here is that the
$\overline{QRM}(1,k+2)$ codes admit the logical $Z_k$ gate
\emph{transversally}, namely by applying $Z_k^\dagger$ to each qubit
independently.  We know this because these codes satisfy the conditions we
derived in Appendix~\ref{sec:main-theorem}.  Because of this transversality
property, we can use the $\overline{QRM}(1, k+2)$ code to distill
$Z_k^\dagger|+\>$ states using circuits that are essentially the same as the
one used in Refs.~\cite{Raussendorf:2007b, Fowler:2012a} to distill
$Z_2|+\>$ states using the 15-qubit code, a circuit that is more compact
than the one originally described by Bravyi and Kitaev \cite{Bravyi:2005a}.
Specifically, if we replace the encoding circuit for $\overline{QRM}(1,4)$
with the encoding circuit for $\overline{QRM}(1,k+2)$ and replace each $T$
with $Z_k$, the circuit becomes a distillation circuit for $Z_k^\dagger|+\>$
states.  Due to the numerical results in Ref.~\cite{Jochym-OConnor:2012a}
that showed that magic states which are left untwirled can still be
distilled, we also omit twirling our bare input states. As an example, we
depict the distillation circuit for $Z_3^\dagger$ in
Fig.~\ref{fig:Zk-distillation}; we derived the encoding circuit for
$\overline{QRM}(1,5)$ in the figure using the methods outlined in
Refs.~\cite{Gottesman:1997a, Nielsen:2000a}.  We defer a proof of why these
codes have the transversality property to Appendix~\ref{sec:main-theorem}
and instead focus on how the protocol works here.  We will note here,
though, that our proof generalizes the ``tri-orthogonality'' condition that
Bravyi and Haah used to establish the transversality of $T$ gates for their
codes to a lemma in coding theory proved by Ward that we call \emph{Ward's
Divisibility Test} \cite{Ward:1990a, Liu:2006a}.

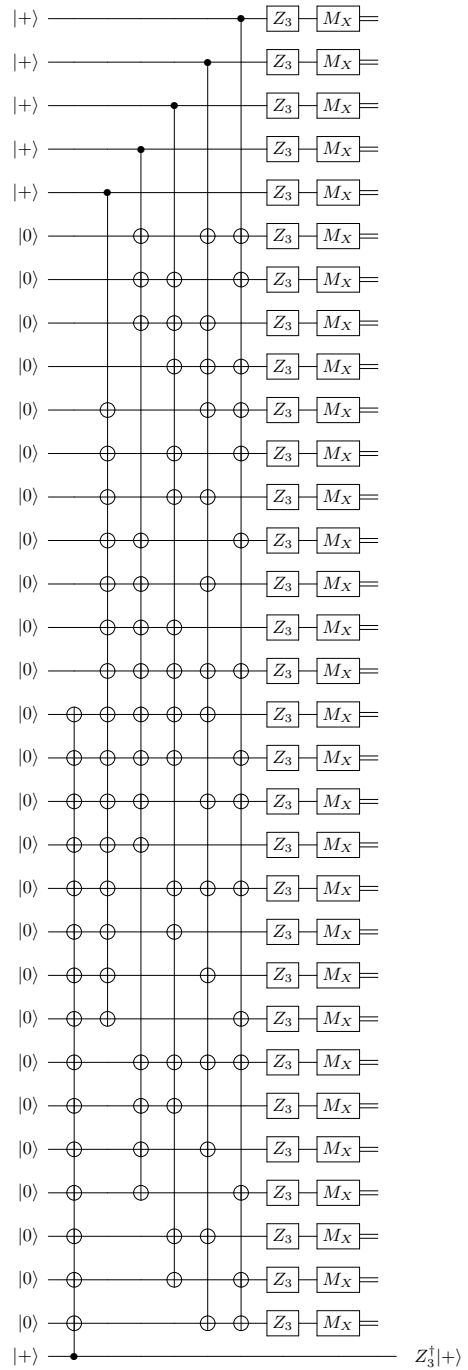
\begin{figure}[h!]
\centerline{
\scalebox{0.75}
{
\Qcircuit @C=1em @R=1em {
\lstick{\ket{+}} &  \qw  &  \qw  &  \qw  &  \qw  &  \qw  &  \ctrl{30}  &  \gate{Z_3} & \gate{M_X} & \cw & & \\
\lstick{\ket{+}} &  \qw  &  \qw  &  \qw  &  \qw  &  \ctrl{29}  &  \qw  &  \gate{Z_3} & \gate{M_X} & \cw & & \\
\lstick{\ket{+}} &  \qw  &  \qw  &  \qw  &  \ctrl{27}  &  \qw  &  \qw  &  \gate{Z_3} & \gate{M_X} & \cw & & \\
\lstick{\ket{+}} &  \qw  &  \qw  &  \ctrl{24}  &  \qw  &  \qw  &  \qw  &  \gate{Z_3} & \gate{M_X} & \cw & & \\
\lstick{\ket{+}} &  \qw  &  \ctrl{19}  &  \qw  &  \qw  &  \qw  &  \qw  &  \gate{Z_3} & \gate{M_X} & \cw & & \\
\lstick{\ket{0}} &  \qw  &  \qw  &  \targ  &  \qw  &  \targ  &  \targ  &  \gate{Z_3} & \gate{M_X} & \cw & & \\
\lstick{\ket{0}} &  \qw  &  \qw  &  \targ  &  \targ  &  \qw  &  \targ  &  \gate{Z_3} & \gate{M_X} & \cw & & \\
\lstick{\ket{0}} &  \qw  &  \qw  &  \targ  &  \targ  &  \targ  &  \qw  &  \gate{Z_3} & \gate{M_X} & \cw & & \\
\lstick{\ket{0}} &  \qw  &  \qw  &  \qw  &  \targ  &  \targ  &  \targ  &  \gate{Z_3} & \gate{M_X} & \cw & & \\
\lstick{\ket{0}} &  \qw  &  \targ  &  \qw  &  \qw  &  \targ  &  \targ  &  \gate{Z_3} & \gate{M_X} & \cw & & \\
\lstick{\ket{0}} &  \qw  &  \targ  &  \qw  &  \targ  &  \qw  &  \targ  &  \gate{Z_3} & \gate{M_X} & \cw & & \\
\lstick{\ket{0}} &  \qw  &  \targ  &  \qw  &  \targ  &  \targ  &  \qw  &  \gate{Z_3} & \gate{M_X} & \cw & & \\
\lstick{\ket{0}} &  \qw  &  \targ  &  \targ  &  \qw  &  \qw  &  \targ  &  \gate{Z_3} & \gate{M_X} & \cw & & \\
\lstick{\ket{0}} &  \qw  &  \targ  &  \targ  &  \qw  &  \targ  &  \qw  &  \gate{Z_3} & \gate{M_X} & \cw & & \\
\lstick{\ket{0}} &  \qw  &  \targ  &  \targ  &  \targ  &  \qw  &  \qw  &  \gate{Z_3} & \gate{M_X} & \cw & & \\
\lstick{\ket{0}} &  \qw  &  \targ  &  \targ  &  \targ  &  \targ  &  \targ  &  \gate{Z_3} & \gate{M_X} & \cw & & \\
\lstick{\ket{0}} &  \targ  &  \targ  &  \targ  &  \targ  &  \targ  &  \qw  &  \gate{Z_3} & \gate{M_X} & \cw & & \\
\lstick{\ket{0}} &  \targ  &  \targ  &  \targ  &  \targ  &  \qw  &  \targ  &  \gate{Z_3} & \gate{M_X} & \cw & & \\
\lstick{\ket{0}} &  \targ  &  \targ  &  \targ  &  \qw  &  \targ  &  \targ  &  \gate{Z_3} & \gate{M_X} & \cw & & \\
\lstick{\ket{0}} &  \targ  &  \targ  &  \targ  &  \qw  &  \qw  &  \qw  &  \gate{Z_3} & \gate{M_X} & \cw & & \\
\lstick{\ket{0}} &  \targ  &  \targ  &  \qw  &  \targ  &  \targ  &  \targ  &  \gate{Z_3} & \gate{M_X} & \cw & & \\
\lstick{\ket{0}} &  \targ  &  \targ  &  \qw  &  \targ  &  \qw  &  \qw  &  \gate{Z_3} & \gate{M_X} & \cw & & \\
\lstick{\ket{0}} &  \targ  &  \targ  &  \qw  &  \qw  &  \targ  &  \qw  &  \gate{Z_3} & \gate{M_X} & \cw & & \\
\lstick{\ket{0}} &  \targ  &  \targ  &  \qw  &  \qw  &  \qw  &  \targ  &  \gate{Z_3} & \gate{M_X} & \cw & & \\
\lstick{\ket{0}} &  \targ  &  \qw  &  \targ  &  \targ  &  \targ  &  \targ  &  \gate{Z_3} & \gate{M_X} & \cw & & \\
\lstick{\ket{0}} &  \targ  &  \qw  &  \targ  &  \targ  &  \qw  &  \qw  &  \gate{Z_3} & \gate{M_X} & \cw & & \\
\lstick{\ket{0}} &  \targ  &  \qw  &  \targ  &  \qw  &  \targ  &  \qw  &  \gate{Z_3} & \gate{M_X} & \cw & & \\
\lstick{\ket{0}} &  \targ  &  \qw  &  \targ  &  \qw  &  \qw  &  \targ  &  \gate{Z_3} & \gate{M_X} & \cw & & \\
\lstick{\ket{0}} &  \targ  &  \qw  &  \qw  &  \targ  &  \targ  &  \qw  &  \gate{Z_3} & \gate{M_X} & \cw & & \\
\lstick{\ket{0}} &  \targ  &  \qw  &  \qw  &  \targ  &  \qw  &  \targ  &  \gate{Z_3} & \gate{M_X} & \cw & & \\
\lstick{\ket{0}} &  \targ  &  \qw  &  \qw  &  \qw  &  \targ  &  \targ  &  \gate{Z_3} & \gate{M_X} & \cw & & \\
\lstick{\ket{+}} & \ctrl{-15} & \qw & \qw & \qw & \qw & \qw & \qw & \qw & \qw & \qw & \rstick{\!\!\!Z_3^\dagger\ket{+}} \\
} 
} 
} 
\caption{\small{\label{fig:Zk-distillation}Distillation circuit for
$Z_3^\dagger|+\> = \sqrt{T}^\dagger|+\>$ states; it is the 31-qubit
shortened quantum Reed-Muller code's encoding circuit applied to half of a Bell
state followed by the logical $Z_3$ gate and $M_X$ measurement of the qubits
on this encoded half.  The $Z_3$ gates are performed using the teleportation
circuit depicted in Fig.~\ref{fig:Zk-teleportation-circuit}.  This circuit
also distills $Z_3|+\>$ states on $Z_3^\dagger|+\>$ inputs.}}
\end{figure}

Using the $\overline{QRM}(1, k+2)$ code to distill $Z_k|+\>$ states yields
the following distillation polynomial:
\begin{align}
\label{eq:CISC-distillation-poly}
\epsilon_{\text{out}}(\epsilon) &=
  \frac{
        1 - (1 - 2\epsilon)^{2^{k+1}-1}
        \left[2\epsilon(2^{k+2} - 1) + (1-2\epsilon)^{2^{k+1}}\right]
       }
       {
        2\left[1 + (2^{k+2}-1)(1-2\epsilon)^{2^{k+1}}\right]
       } \\
  &\approx \left(1 - 3\cdot 2^{k+1} + 2^{2k+3}\right)(\epsilon^3/3 +
\epsilon^4 + \bigO(\epsilon^5)).
\end{align}  
Approximate values for the distillation threshold for various values of $k$
are listed in Table~\ref{tab:Zk-distillation-polynomials-table}; these are
the same threshold values one would have obtained if one had used the code
for distilling $Z_k|+\>$ to distill $Z_{k+1}|+\>$, but the improvement in
accuracy in such a case would only be to $\bigO(\epsilon)$ instead of
$\bigO(\epsilon^3)$ by generalizing the method of Reichardt
\cite{Reichardt:2005a}.

\begin{table}[ht!]
\centering
\begin{tabular}{c|r|r} \hline \hline
$k$ & $\epsilon_{\text{out}}/\epsilon^3$ & $\epsilon_k^{\text{th}}$ \\ \hline
 2 & 35  & 14.15\% \\
 3 & 155 & 6.94\% \\
 4 & 651 & 3.44\% \\
 5 & 2\,667 & 1.71\% \\
 6 & 10\,795 & 0.85\% \\
 7 & 43\,435 & 0.43\% \\
 8 & 174\,251 & 0.21\% \\
 9 & 698\,027 & 0.11\% \\
 10 & 2\,794\,155 & 0.05\% \\
\hline \hline
\end{tabular}
\caption{Distillation polynomials (to most significant order) and
distillation thresholds for distilling $Z_k^\dagger|+\>$ states.}
\label{tab:Zk-distillation-polynomials-table}
\end{table}

Although the distillation threshold drops as $k$ increases, it is still
larger than or comparable to the threshold of $\approx 1\%$ for
fault-tolerant quantum computation with surface codes \cite{Dennis:2002a,
Raussendorf:2007b, Fowler:2008a} for values of $k$ less than or equal to
$6$, where it takes the value $\epsilon_6^{\text{th}} \approx 0.85\%$.  This
then sets a reasonable upper limit on the size of the complex instruction
set one should consider for performing $Z_k$ gates in this way; going
further would place greater fidelity demands on the elementary operations
than fault-tolerance does.

To achieve $\epsilon_{\text{out}} \leq \epsilon'$, one must iterate the
distillation circuit
\begin{align}
\label{eq:CISC-ell-iterations}
\ell(\epsilon', \epsilon) &=
  \left\lceil \frac{\log \epsilon'}{\log
\epsilon_{\text{out}}(\epsilon)}\right\rceil
\end{align}
times.  The expected number of repetitions per iteration needed to achieve
distillation success is
\begin{align}
\label{eq:CISC-expected-t-repetitions}
\mathds{E}[t(\epsilon)] 
  &= \frac{2^{k+2}}{1 + (2^{k+2}-1)(1-2\epsilon)^{2^{k+1}}}.
\end{align}

Unlike in the \textsc{risc} protocol, in which the corrective step in the
teleportation circuit added no error, in our protocol each teleportation
circuit may add error in its adaptive $Z_{k-1}$ gate. Therefore, we must
implement the $Z_{k-1}$ gate with low error using our protocol recursively.
We require that the error in the corrective $Z_{k-1}$ gate be \emph{at most}
the error in the $Z_{k}$ gates in Fig. \ref{fig:Zk-distillation}. Due to the
differences in the distillation polynomials for different values of $k$, it
turns out that the error in the $Z_{k-1}$ gates for the corrective step is
always less than the error in the $Z_k$ gates as long as both are being
implemented by magic states that have been subjected to the same number of
levels of distillation using our protocols.

%
\subsection{Resource analysis}

Asymptotically, our \textsc{cisc} protocol achieves a value of $\beta =
\beta_k := \log_3(2^{k+2}-1)$ and $\gamma = 0$.  The sum $\beta + \gamma$ is
less than the sum of the 15-to-1 Bravyi-Kitaev magic-state distillation
$\beta$ and the Dawson-Nielsen compiling $\gamma$ for $k \leq 9$.  However,
since the distillation threshold drops below $0.85\%$ after $k = 6$, as
argued earlier, it is probably wisest to stop at $k = 6$.  Compared to the
best values we know for $\beta$ ($\approx 1.58$ by Refs.~\cite{Bravyi:2012a,
Jones:2012a}) and $\gamma$ ($1$ by Ref.~\cite{Kliuchnikov:2012b}), our
\textsc{cisc} protocol would appear to be only superior for $k \leq 2$.
However it is important to remember, as mentioned earlier, that arguing
about asymptotics in this way can be very misleading as the constants
involved can be huge.  Indeed, asymptotically our protocol is inferior in
that it requires many more resource states than the Selinger $+$ MEK scheme.
However, we find that for a fairly long range of values of $\epsilon'$ and
$k$, our protocol performs better, not becoming worse until $\epsilon'
\approx 10^{-10}$ for $k=5$ and $k=6$ and staying comparable or better for
$k=3$ and $k=4$ to accuracies of $\epsilon' < 10^{-70}$. Due to the discrete
jumps taken in the resource requirements of our protocol, the precise
analysis becomes a bit subtle. The plot in
Fig.~\ref{fig:RISC-v-CISC-Zk-plot} gives a better feel for when it is
favorable to use our \textsc{cisc} protocol.

An important difference in accounting for the resource demands of our
protocol as compared to the \textsc{risc} solution is that, while we incur
no overhead from quantum compiling, we do have a potentially more resource
intensive teleportation step. While in the \textsc{risc} protocol the
eventual use of a distilled magic state required only a possible Clifford
correction in the teleportation procedure, in the \textsc{cisc} protocol we
have to also account for the fact that when teleporting a $Z_k|+\>$ state it
may be necessary to perform a $Z_{k-1}$ correction that is accurate to
\emph{at least} the same $\epsilon'$.

For the \textsc{cisc} architecture, we only allow ourselves access to
$Z_k|+\>$ states of precision $\epsilon$ and the use of QRM-based
distillation routines, even for $k=2$. Because of this, we slightly over
count the resources required by not optimizing over the best routine to
produce a $Z_2|+\>$ state of a desired $\epsilon'$.

We produce our counts via the following recursive formula:
\begin{align}
n_{\text{states}}^{\textsc{cisc}}(k, \ell) &=
\left(2^{k+2}-1\right)\bigg[n_{\text{states}} ^{\textsc{cisc}}(k, \ell-1) \\
\nonumber
&+ \frac{1}{2} n_{\text{states}}^{\textsc{cisc}}(k-1, \ell-1)\bigg] \cdot
\mathds{E}[t(\epsilon)\\
\nonumber
&+ \frac{1}{2} n_{\text{states}}^{\textsc{cisc}}(k-1, \ell),
\end{align}
where the base of the recursion is given by
\begin{gather}
n_{\text{states}}^{\textsc{cisc}}(2, \ell) = \mathds{E}[t(\epsilon)]
15^\ell.
\end{gather}
The factor $\mathds{E}[t(\epsilon)]$, which accounts for the need to repeat
the protocol if an improper measurement outcome is obtained, is very nearly
$1$ for the first level of distillation given bare states of accuracy
$\epsilon = 10^{-4}$, and is even closer to $1$ at higher levels when the
input states are accurate to even higher precision. The leading $2^{k+2} -
1$ is due to the number of $Z_k|+\>$ states needed at each level $\ell$ of
distillation. The first term in the square brackets accounts for the fact
that distilling a new state at level $\ell$ requires states already
distilled to level $\ell-1$, while the second term accounts for the fact
that each of these $Z_k|+\>$ states from level $\ell-1$ are injected to our
circuit via teleportation and on average half will require a $Z_{k-1}$
correction, also from distillation level $\ell-1$. The final term counts the
resources needed for the final teleportation step that consumes the
distilled magic state. Here, half the time we will need to perform a
$Z_{k-1}$ correction which must be distilled to the same level as the $Z_k$
gate being applied.

%
\section{Conclusions}
\label{sec:conclusion}

Fig.~\ref{fig:RISC-v-CISC-Zk-plot} shows the results of counting resource
states for the various protocols we've described. Interpreting the results
is subtle, with our protocols performing better when using only one or two
rounds of distillation and losing out later as the asymptotics take over. As
mentioned earlier, our protocols are asymptotically much worse that the
current state of the art, but for accuracies of $\epsilon' > 10^{-10}$, or
indeed as low as $10^{-70}$ for $k=3$ or $k=4$, the \textsc{cisc} solution
outperforms the \textsc{risc} solution. Some of the \textsc{cisc} protocols
show an interesting reentrant behavior, becoming better than the
\textsc{risc} protocol as accuracy demands increase even though they started
out using more states at lower accuracies. This is due to the large steps in
accuracy when another level of distillation is used in our scheme.

\begin{figure}[h!]
\center{\includegraphics[width=1.0\columnwidth]{\FigDirectory/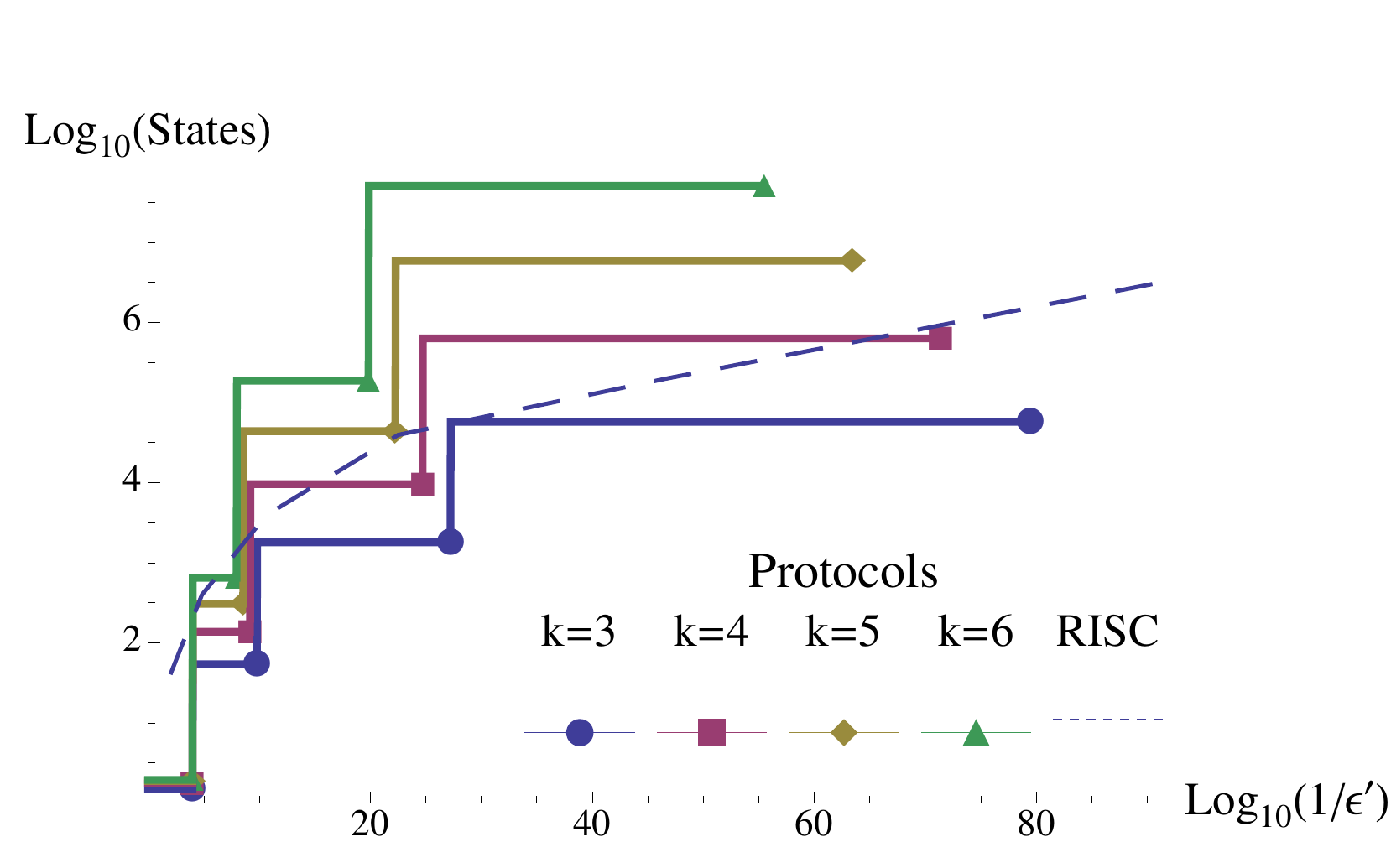}}
\caption{\small{\label{fig:RISC-v-CISC-Zk-plot}Log of the number of resource states
required to synthesize the quantum $Z(\pi/2^k)$ gate as a function of the
log of the inverse of the desired precision $\epsilon'$ for the
\textsc{risc} architecture described in the text and our \textsc{cisc}
architecture.}}
\end{figure}

The difference between the architectures at low precision demand
reflects the fact that when the hardware error rate is already below this
demand (\ie, when $\epsilon < \epsilon'$), the only gates required by our
quantum \textsc{cisc} architecture are those used to teleport the gate $Z_k$
from the state $Z_k|+\>$ to the target state $|\psi\>$.  The \textsc{risc}
architecture doesn't include the $Z_k$ gate for $k > 2$, so it must instead
use a quantum compiling strategy to synthesize $Z_k$ from $T|+\>$ states.

Our \textsc{cisc} architecture does have some limitations.  To begin, as can
be seen in Fig.~\ref{fig:RISC-v-CISC-Zk-plot}, as $k$ increases, even at
fixed precision demand $\epsilon'$, the number of gates our \textsc{cisc}
architecture uses increases.  At any fixed $\epsilon'$, even those
corresponding to very low accuracies, there will be some $k$ for which the
\textsc{risc} architecture uses fewer gates.  However, a feature not
apparent in this plot but apparent from
Table~\ref{tab:Zk-distillation-polynomials-table} is that, even before this
happens, the distillation threshold for our \textsc{cisc} architecture drops
to a point below the accuracy threshold for fault-tolerant quantum
computation.  Using our \textsc{cisc} architecture beyond $k=6$ would be
foolhardy, as suddenly the distillation of encoded instructions and not the
capacity of the underlying code would set the experimental hardware demands
at the physical level.  For this reason, we advocate using our \textsc{cisc}
architecture up to $k=6$, and then relying on an external quantum compiling
algorithm (but with a larger base instruction set than a quantum
\textsc{risc} architecture would have) to synthesize $Z_k$ rotations for
larger $k$ values.

We focused on synthesizing $Z_k$ rotations for two reasons.  First, numerous
quantum algorithms rely on the quantum Fourier transform, which in turn is
naturally decomposed into Clifford operations and $Z_k$ rotations.  We
thought it was important to focus on synthesizing transformations that arise
in actual algorithms rather than operations that occur only in the abstract.
Second, and more significantly, we were able to find a code family, the
shortened quantum Reed-Muller codes, that we could leverage to create
distillation protocols for $Z_k$ rotations.  The key enabling property these
codes possess is \emph{code divisibility}.  With this insight, we
generalized the ``tri-orthogonality'' condition of Bravyi and Haah
\cite{Bravyi:2012a} to a condition we call Ward's Divisibility Test, which
recognizes its analogous role in classical coding theory \cite{Ward:1990a}.
We haven't sought codes beyond the shortened quantum Reed-Muller codes that
pass Ward's Divisibility Test for admitting a $Z_k$-distillation protocol.
However, we present and prove the correctness of this test in
Appendix~\ref{sec:main-theorem} in the hopes that others will find it
helpful in the quest to improve quantum \textsc{cisc} architectures.  

One of the overall messages of our work is that it is not optimal to first
optimize the number of gates used to synthesize a universal instruction set
and then optimize the number of universal instructions needed to synthesize
a gate of interest, in this case, a $Z_k$ gate.  Instead, one can reap
significant advantages by approaching this as a single optimization problem.
The best conjectured asymptotic scaling when approached as two separate
problems requires a number of gates that scales as
$\bigO(\log^2(1/\epsilon'))$.  By approaching this as a single optimization
problem, one may be able to achieve $\bigO(\log(1/\epsilon'))$ for the
combined process.

The resource tradeoff space for implementing quantum operations with finite
discrete instruction sets is an area ripe for investigation.  Beyond just
minimizing the number of resource states required to approximate
transformations of interest (our focus here), one might be interested in
minimizing other metrics, such as the number of gates, the number of qubits
used, the depth of the approximating quantum circuit, or the size of the
approximating quantum circuit (which is its depth times the number of
qubits).  Depending on the task at hand, one instruction set may be more
suitable than another.  Investigations along these lines help us better
understand the limits and capabilities of finite-instruction-set quantum
information processing.

%
\begin{acknowledgments}

AJL and CC were supported in part by the Laboratory Directed Research and
Development program at Sandia National Laboratories.  Sandia National
Laboratories is a multi-program laboratory managed and operated by Sandia
Corporation, a wholly owned subsidiary of Lockheed Martin Corporation, for
the U.S.  Department of Energy's  National Nuclear Security Administration
under contract DE-AC04-94AL85000.

\end{acknowledgments}

%

\begin{appendix}

%
\section{Quantum Reed-Muller codes}
\label{sec:Quantum-Reed-Muller-codes}

One of the challenges in discussing quantum Reed-Muller codes is that there
is not a unique definition of what a quantum Reed-Muller code is in the
literature \cite{Steane:1996d, Zhang:1997a, Preskill:1998b, Bravyi:2005a,
Sarvepalli:2005a, Campbell:2012a}.  Fortunately, there is at least a
well-established definition for what a classical Reed-Muller code is.  We
state the definition for classical Reed-Muller codes below, confining our
attention to binary codes.  We refer the reader to standard texts for the
definitions of supporting concepts such as Boolean monomials and $GF(2)$
\cite{MacWilliams:1977a}.

\begin{definition}
The \emph{$r$th-order binary Reed-Muller code of length $2^m$}, denoted
$RM(r,m)$, is the linear code over $GF(2)$ whose generator matrix is
composed of row vectors corresponding to the Boolean monomials over
$GF(2)^{2^m}$ of degree at most $r$.
\end{definition}

As an example, the generator matrix for the $RM(1,4)$ code is
\setcounter{MaxMatrixCols}{32}
\begin{align}
\label{eq:G-matrix}
 G = 
  \begin{bmatrix}
    1 & 1 & 1 & 1 & 1 & 1 & 1 & 1 & 1 & 1 & 1 & 1 & 1 & 1 & 1 & 1 \\
    1 & 1 & 1 & 1 & 1 & 1 & 1 & 1 & 0 & 0 & 0 & 0 & 0 & 0 & 0 & 0 \\
    1 & 1 & 1 & 1 & 0 & 0 & 0 & 0 & 1 & 1 & 1 & 1 & 0 & 0 & 0 & 0 \\
    1 & 1 & 0 & 0 & 1 & 1 & 0 & 0 & 1 & 1 & 0 & 0 & 1 & 1 & 0 & 0 \\
    1 & 0 & 1 & 0 & 1 & 0 & 1 & 0 & 1 & 0 & 1 & 0 & 1 & 0 & 1 & 0 
  \end{bmatrix}.
\end{align}

From this definition, the codespace of binary Reed-Muller codes is just the
space of Boolean polynomials over $GF(2)^{2^m}$ of degree at most $r$.  It
is a minor combinatoric exercise to work out that the code $RM(r,m)$ has
rank $k = \sum_{i=0}^{r} \binom{m}{i}$ and code distance $d = 2^{m-r}$.  In
standard coding theory notation, we say that the code $RM(r,m)$ is an
\begin{align}
[n, k, d] = \left[2^m,\ \ \sum_{i=0}^r \tbinom{m}{i},\ \ 2^{m-r}\right]
\end{align}
code.

It is straightforward to work out that the dual code to $RM(r,m)$ is
$RM(m-r-1,m)$.  We use this to define a quantum Reed-Muller code as a CSS
code composed of $RM(r,m)$ and its dual:
\begin{definition}
The \emph{$r$th-order quantum binary Reed-Muller code of length $2^m$},
denoted $QRM(r,m)$, is the CSS code \cite{Calderbank:1996a, Steane:1996c}
whose defining $X$ and $Z$ parity check matrices are the generator matrices
for $RM(r,m)$ and its dual $RM(m-r-1, m)$ respectively.
\end{definition}
Notice that in this definition, somewhat confusingly, the quantum
parity-check matrices are formed from classical \emph{generator} matrices,
not classical parity-check matrices.

We are most interested in the \emph{shortened} quantum binary Reed-Muller
codes, which we denote by $\overline{QRM}(r,m)$.  These codes are formed by
shortening each of the binary Reed-Muller codes from which it is formed.
The process of shortening first punctures a code by removing a bit on which
only row of the generator matrix has support and then expurgates it by
removing the row in the generator matrix that had support on that bit.  For
the Reed-Muller codes, this corresponds to removing the first row and last
column of the generator matrix when presented in standard form, as in
Eq.~(\ref{eq:G-matrix}).  In essence, shortening a Reed-Muller code
restricts the space of Boolean polynomials defining the code to those which
have no constant term and which also satisfy $p(0) = 0$.  An equivalent way
of characterizing the shortened Reed-Muller code is as the even subcode of
the punctured Reed-Muller code.  The parameters of the resulting quantum
code are $[\![2^m - 1, 1]\!]$.  Code parameters for small Reed-Muller codes,
their duals, and their shortened quantum construct are listed in Table
\ref{tab:RM-codes-and-duals}.  Notice that the length of the code $n$ does
not uniquely specify which shortened quantum Reed-Muller code one is
referring to for $n > 15$.

\begin{table}[h!]
%
%
\begin{tabular}{c|c|c|c|c} \hline \hline
$(r,m)$ & $(m-r-1,m)$ & $[n, k, d]$ primal & $[n, k, d]$ dual & $[\![n,
k]\!]$\\ \hline
 (0,1) & (0,1) & [2,1,2] & [2,1,2] & $\varnothing$ \\
 (0,2) & (1,2) & [4,1,4] & [4,3,2] & $\varnothing$ \\
 (0,3) & (2,3) & [8,1,8] & [8,7,2] & $\varnothing$ \\
 (1,3) & (1,3) & [8,4,4] & [8,4,4] & $[\![7, 1]\!]$ \\
 (0,4) & (3,4) & [16,1,16] & [16,15,2] & $\varnothing$ \\
 (1,4) & (2,4) & [16,5,8] & [16,11,4] & $[\![15, 1]\!]$ \\
 (0,5) & (4,5) & [32,1,32] & [32,31,2] & $\varnothing$ \\
 (1,5) & (3,5) & [32,6,16] & [32,26,4] & $[\![31, 1]\!]$ \\
 (2,5) & (2,5) & [32,16,8] & [32,16,8] & $[\![31, 1]\!]$ \\
 (0,6) & (5,6) & [64,1,64] & [64,63,2] & $\varnothing$ \\
 (1,6) & (4,6) & [64,7,32] & [64,57,4] & $[\![63, 1]\!]$ \\
 (2,6) & (3,6) & [64,22,32] & [64,42,8] & $[\![63, 1]\!]$ \\
\hline \hline
\end{tabular}
\caption{Parameters for (primal) Reed-Muller $R(r,m)$ codes, their duals
$R(m-r-1,1)$, and their CSS-combined shortened quantum versions
$\overline{QRM}(r,m)$ for small values.  Shortened $R(0,m)$ codes have no
$X$ generator, so the resulting quantum codes are just classical codes; they are
referred to by $\varnothing$ in the table.}
\label{tab:RM-codes-and-duals}
\end{table}

%
\section{Criteria for a code to admit transversal $Z(\pi/2^k)$ rotations}
\label{sec:main-theorem}

The shortened quantum Reed-Muller codes $\overline{QRM}(1,k+2)$ admit a
transversal implementation of $Z_k$ by applying $Z_k^\dagger$ to each qubit
in the code independently.  This result follows, \eg\, from arguments made
by Campbell \etal\ in Ref.~\cite{Campbell:2012a}.  Another way to see this
is to note that these codes obey Theorem \ref{thm:Zk-transversality} below.
We offer this alternative approach because it may be generalizable in a way
that others could use to find more efficient codes that admit $Z_k$
transversally.  It also relies on a lemma (Lemma \ref{lem:Ward}) that
naturally generalizes an otherwise unusual criterion of
``tri-orthogonality'' noted by Bravyi and Haah \cite{Bravyi:2012a} for the
$\overline{QRM}(1,4)$ code.  We believe that this Lemma, which we call
\emph{Ward's Divisibility Test}, makes better contact with the classical
coding theory literature.

\begin{theorem}
\label{thm:Zk-transversality}
A quantum $[\![n, 1]\!]$ CSS code \cite{Calderbank:1996a, Steane:1996c} with
stabilizer generators defined by the parity check matrix $H = \textrm{diag}(H^X,
H^Z)$ via
\begin{align}
S_i^X &:= \bigotimes_{j = 1}^n X^{H_{ij}^X} 
& S_i^Z &:= \bigotimes_{j = 1}^n Z^{H_{ij}^Z},
\end{align}
where $H^X$ has rows $v_1, \ldots v_{k+2}$, implements $\left(Z_k\right)^a$
transversally if
\begin{align}
\label{eq:k-plus-one-orthogonality}
\textrm{wt}\left(v_{\sigma(1)} \cdots v_{\sigma(j)}\right)
 &\equiv 0 \bmod 2^{k+2-j}
\end{align}
for all $1 \leq j \leq k+2$ and all $\sigma \in \mathit{\Sigma}_{j}$, and 
\begin{align}
\label{eq:n-equals-a}
  n &\equiv a \bmod 2^{k+1},
\end{align} 
where `$\otimes$' denotes the tensor product, `$\textrm{wt}$' denotes the
Hamming weight of a binary vector, `$\mathit{\Sigma}_j$' denotes the
permutation group on $j$ items, and `$v_1 \cdots v_j$' denotes the
componentwise product of $v_1, \ldots, v_j$.
\end{theorem}

When $a$ in this Theorem is odd, $\gcd(a,2^{k+1}) = 1$, which means we can
use an algorithm like the extended Euclidean algorithm \cite{Cormen:2001a}
to efficiently find numbers $x$ and $y$ such that $ax + 2^{k+1}y = 1$.
Iterating $(Z_k)^a$ $x$ times results in a conditional phase of $\pi(1 -
2^{k+1}y)/2^{k} \cong \pi/2^k$; in other words, $(Z_k)^{ax} \cong Z_k$ when
$a$ is odd.

Condition (\ref{eq:k-plus-one-orthogonality}) generalizes the
tri-orthogonality condition of Bravyi and Haah \cite{Bravyi:2012a} into a
kind of $(k+1)$-orthogonality condition.  More fundamentally, we want the
classical linear code generated by $H^X$ to be a code in which every
codeword has a Hamming weight divisible by $2^{k+1}$.  Ward studied such
\emph{divisible codes} in depth and one of  his results is that
$2^{k+1}$-divisibility is testable by the condition of Eq.
(\ref{eq:k-plus-one-orthogonality}) \cite{Ward:1990a}.  More explicitly,
Ward's Divisibility Test is captured by Lemma \ref{lem:Ward} below.  (Ward's
result is actually more general; we use a version specialized to the binary
case, as noted by Proposition $4.2$ in Ref.~\cite{Liu:2006a}.)

\begin{lemma}[Ward's Divisibility Test \cite{Ward:1990a}]
\label{lem:Ward}
The binary linear code with generator matrix $H^X$ whose row vectors are
$v_1, \ldots, v_{k+2}$ is divisible by $2^{k+1}$ if and only if 
\begin{align}
2^{k+2-j} \big| \textrm{wt}(v_{\sigma(1)} \cdots v_{\sigma(j)})
\end{align}
for all $1 \leq j \leq k+1$ and all permutations $\sigma \in \mathit{\Sigma}_j$.
\end{lemma}

While Ward's Divisibility Test has the advantage of being an explicit
algorithm for testing divisibility, it is not particularly efficient, as it
takes a time that is exponential in $k$ to execute.  For codes with a high
degree of structure, such as the shortened $\overline{RM}(1,k+2)$
Reed-Muller codes, demonstrating $2^{k+1}$ divisibility is much simpler, as
noted in Ref.~\cite{Liu:2006a}.

\begin{proof}[Proof of Theorem \ref{thm:Zk-transversality}]
By Ward's Divisibility Test, every vector $v$ in the rowspan $\calL$ of
$H^X$ has a Hamming weight divisible by $2^{k+1}$.  Since the logical $|0\>$
for the code is $|\overline{0}\> := \sum_{v \in \calL}|v\>$ (ignoring
normalization), the action of transversal $Z_k$ on $|\overline{0}\>$ is
\begin{align}
Z_k^{\otimes n} |\overline{0}\>
  &= \sum_{v\in\calL} Z_k^{\otimes n} |v\> \\
  &= \sum_{v\in\calL} \left(e^{i \pi/2^k}\right)^{\left|v\right|} |v\> \\
  &= \sum_{v \in \mathcal{L}} |v\> \\
  &= |\overline{0}\>.
\end{align}
Similarly, using Eq.~(\ref{eq:n-equals-a}), the action of transversal $Z_k$
on (unnormalized) $|\overline{1}\> = \overline{X}|\overline{0}\>$ is
\begin{align}
Z_k^{\otimes n} |\overline{1}\>
  &= Z_k^{\otimes n} \overline{X}|\overline{0}\> \\
  &= \sum_{v\in\calL} Z_k^{\otimes n} \overline{X} |v\> \\
  &= \sum_{v\in\calL} Z_k^{\otimes n} |v \oplus \mathbf{1}\> \\
  &= \sum_{v\in\calL} \left(e^{i \pi/2^k}\right)^{n - \left|v\right|}
      |v \oplus \mathbf{1}\> \\
  &= \sum_{v\in\calL} \left(e^{i \pi a/2^k}\right) |v \oplus \mathbf{1}\> \\
  &= e^{i \pi a/2^k} |\overline{1}\>,
\end{align}
where $\mathbf{1}:= (1, \ldots, 1)$ denotes the all-ones vector, whose
appearance comes from the fact that up to local qubit basis changes,
$\overline{X} = X^{\otimes n}$ for all CSS codes.  These actions of
$Z^{\otimes n}$ replicate $(Z_k)^a$ on the logical basis, and therefore
$Z_k$ implements $(Z_k)^a$ transversally.
\end{proof}

\end{appendix}

\bibliographystyle{landahl}
\bibliography{landahl.JabRef}

\end{document}